\RequirePackage[l2tabu,orthodox]{nag}
\documentclass
[11pt,letterpaper]
{article} 

\usepackage[utf8]{inputenc}
\usepackage{amsmath}
\usepackage{amsthm}
\usepackage{amsfonts}
\usepackage{amssymb}
\usepackage{pgfplots}
\usepackage{systeme}
\usepackage{pgf,tikz}
\usepackage{mathrsfs}
\usepackage{extarrows}
\usetikzlibrary{arrows}
\usepackage{tikz-3dplot}
\usepackage{pgfplots}
\pgfplotsset{compat=1.8}
\usepackage[parfill]{parskip}
\usepackage{xcolor}
\usepackage{bm}
\usepackage{turnstile}

\usepackage{hyperref}
\usepackage[capitalise,nameinlink]{cleveref}
\crefname{algocf}{Algorithm}{Algorithms}
\Crefname{algocf}{Algorithm}{Algorithms}

\usepackage[
letterpaper,
top=1.2in,
bottom=1.2in,
left=1in,
right=1in]{geometry}

\DeclareMathOperator*{\argmax}{\arg\!\max}

\newcommand{\matbeg}{\left(\begin{array}}
\newcommand{\matend}{\end{array}\right)}

\usepackage{eurosym}
\DeclareRobustCommand{\officialeuro}{%
  \ifmmode\expandafter\text\fi
  
{\fontencoding{U}\fontfamily{eurosym}\selectfont e}}

\usepackage{mdframed}

\providecommand*{\sststile}[2]{\vdash^{\raise{4mu}{\mkern-8mu #2}}_{\raise{-4mu}{\mkern-8mu #1}}}

\providecommand*{\R}{\mathbb R}

\renewcommand*{\deg}{\mathrm{deg}}

\providecommand*{\Tr}{\mathop{\mathrm{Tr}}}
\providecommand*{\diag}{\mathop{\mathrm{diag}}}
\providecommand*{\E}{\mathop{\mathbb{E}}}

\providecommand*{\bm}[1]{\boldsymbol{#1}}

\renewcommand{\Pr}{\mathbb P}

\renewcommand{\epsilon}{\varepsilon}

\newtheorem{theorem}{Theorem}[section]

\newtheorem{lemma}[theorem]{Lemma}
\newtheorem{corollary}[theorem]{Corollary}
\newtheorem{claim}[theorem]{Claim}
\newtheorem{remark}[theorem]{Remark}
\newtheorem{fact}[theorem]{Fact}
\newtheorem{proposition}[theorem]{Proposition}

\newcommand{\one}{\mathbf 1}
\newcommand{\X}{X_\ell}

\usepackage{fancyhdr}

\begin{document}

\pagestyle{plain}
\setcounter{page}{1}

\title{Kikuchi Graphs of Random Hypergraphs are Approximately Johnson}
\author{Pravesh K. Kothari \thanks{Princeton University}}
\maketitle
\begin{abstract}
We prove that level-$\ell$ Kikuchi graphs of random $2r$-uniform hypergraphs spectrally approximate Kikuchi graph of the complete $2r$-uniform hypergraph at a sampling rate that is sharp up to a logarithmic factor, in the regime $r\leq \ell \leq n/2$. Our proof is based on the matrix Bernstein inequality, but, unlike prior works, we apply it to an appropriate collection of blocks of Johnson eigenspaces. Our analysis relies on a new, simple band-locality property for arbitrary Kikuchi graphs. As an application, we prove that the natural degree-$2\ell$ sum-of-squares relaxation for the Max $2r$-XOR problem is ``integral'' when the input is a planted noisy $2r$-XOR instance on a random hypergraph with $\gtrsim n \cdot (n/\ell)^{r-1} \log n$ hyperedges. 
\end{abstract} 

\newpage
\newcommand{\ba}{\bar{\alpha}}
\newcommand{\bb}{\bar{\beta}}
\newcommand{\bx}{\bar{s}}
\newcommand{\bz}{\bar{z}}
\newcommand{\by}{\bar{t}}
\newcommand{\1}{\bm{1}}
\newcommand{\Hf}{H_{\mathsf{full}}}
\newcommand{\Span}{\mathsf{span}}

\section{Introduction}
Let $H$ be a $2r$-uniform hypergraph on $[n]$ for a positive integer $r$. The level $\ell$-Kikuchi graph of $H$, denoted by $K^{(\ell)}_H$, is the graph on ${[n] \choose \ell}$ where two vertices $S,T$ are connected iff $S \Delta T \in H$. 

Kikuchi graphs were introduced independently in closely related forms in~\cite{WeinAM19,Hastings20TensorPCA}, and various variants have since then been used in the proof of several recent results in coding theory and combinatorics~\cite{GuruswamiKM22,HsiehKM23,GHKM23,AlrabiahGKM22,KothariManohar24LinearLCC,KothariManohar24SmoothLCC,HsiehKothariMohantyMunhaSudakov24,Yankovitz24Linear3LCC,BasuHsiehKothariLin25OddLDC}. When $H$ is a graph (i.e., $r=1$), Kikuchi graphs (sometimes called \emph{token graphs}) also appear in combinatorics with applications to quantum constraint satisfaction~\cite{FabilaMonroyFFHUW12,ApteParekhSud25,BakshiBasuKothariLi26,Lew26ApproxBrouwer,DalfoDFFHTZ21,Lew24TokenGraphs,DalfoFM25,ReyesDalfoFiolMessegue23}.

We prove that Kikuchi graphs of random hypergraphs spectrally approximate, or \emph{sparsify}, the Kikuchi graphs of the complete hypergraph. To describe our result more precisely, let us introduce some setup and notation. We write $\gtrsim_r$ and $\lesssim_r$ to denote inequalities that hold up to multiplicative constants depending only on $r$ (without the subscript, the symbols denote inequalities up to an absolute multiplicative constant slack). We say that $H$ is a random hypergraph with \emph{density} $p$ if every hyperedge in ${n \choose 2r}$ is included in $H$ with probability $p$ independently. We write $\Hf$ to denote the \emph{complete} $2r$-uniform hypergraph on $[n]$. Note that $K^{(\ell)}_{\Hf}$ is a graph from the \emph{Johnson scheme} (in this paper, we will call it the Johnson graph, in short) where an edge $\{S,T\}$ exists whenever $|S \Delta T| = 2r$ or, equivalently, $|S \cap T| = \ell - r$. Let $L_H^{(\ell)}$ be the Laplacian of $K^{(\ell)}_H$ and $L_{\Hf}^{(\ell)}$ be the Laplacian of $K^{(\ell)}_{\Hf}$, respectively.

\begin{theorem}[Random Hypergraphs are Kikuchi Sparsifiers] \label{thm:main}
Let $r \leq \ell \leq n/2$, let $0<\epsilon\leq 1/2$, and let $H \subseteq {[n] \choose 2r}$ be a random $2r$-uniform hypergraph with density $p$ satisfying $p {n \choose 2r} \gtrsim_{r} (n/\epsilon^2) (n/\ell)^{r-1} \log n$. Then, with probability at least $0.99$ over the draw of $H$, 
\[
(1-\epsilon) p L_{\Hf} \preceq L_{H} \preceq (1+\epsilon) p L_{\Hf}.
\]
\end{theorem}

\paragraph{Concrete instantiations:} 
When $r = 1$ (i.e., $H$ is a graph) and $\ell = 1$, observe that $K^{(\ell)}_H = H$, the graph itself. In this setting, the theorem says that random hypergraphs spectrally approximate the complete graph, provided that the expected average degree of a vertex is $\Omega(\log n)$.  More generally, when $\ell =r$, the theorem implies the less standard but still well known (in the context of refuting random XOR formulas~\cite{CojaGL07,GoerdtKrivelevich01}) observation that the natural flattening of $H$ as a ${n \choose r}\times {n \choose r}$ matrix is a graph that spectrally approximates the complete graph so long as the number of hyperedges is $\Omega(n^r \log n)$ (equivalently, the flattened graph on ${n \choose r}$ has $\Omega(\log n)$ degree). 

When $\ell > r$, Kikuchi graphs are ``lifts" of the underlying graphs/hypergraphs. 
For the threshold value of the sampling rate $p$ in the statement, the average degree of the Kikuchi graph is easily computed to be $\Theta(\ell \log n) \approx_{\log n} \log N$ for $N = {n \choose \ell}$. Thus, the statement is equivalent, up to a logarithmic factor, to the claim that Kikuchi graphs of complete hypergraphs are well approximated by Kikuchi graphs of random hypergraphs whenever their expected average degree is $\Omega(\log N)$.

\paragraph{The density requirement is sharp up to a $\log n$ factor for constant $\epsilon$:} We give a short argument here. For every $\ell < n/100$, with probability at least $0.99$ a random $H$ with density $p \gg n/{n \choose 2r}$ must contain $c_r \ell$ disjoint hyperedges for some constant $c_r>0$ depending only on $r$. Let $S$ be a subset of size $\ell$ that intersects each such edge in exactly $r$ elements. Then, $\deg(S) > c_r \ell$ in which case $\|L_H\|_2 \geq c_r \ell$. On the other hand, observe that $\|L_{\Hf}\|_2 \leq 2{\ell \choose r} {n \choose r}$ since the degree of every vertex of $K_{\Hf}$ is upper bounded by ${\ell \choose r} {n \choose r}$. If $p {n \choose 2r} \lesssim_{r} n (n/\ell)^{r-1}$, we can ensure that $\|p L_{\Hf}\|_2 < c_r \ell$. So the inequality in the statement above cannot hold for any fixed $\epsilon<1$. 

\paragraph{Absolute vs Relative Error Bounds:} 
In the context of refuting random constraint satisfaction problems, prior work~\cite{WeinAM19,Hastings20TensorPCA} proves that $\|\sum_{C \in H} b_C K_C^{(\ell)}\|_2 \lesssim \epsilon p {\ell \choose r}{n \choose r}$ whenever $p$ satisfies the same lower bound condition as in \Cref{thm:main}. Later works~\cite{GuruswamiKM22,HsiehKM23}  that generalize such results to the harsher \emph{semirandom} setting prove similar bounds for arbitrary hypergraphs (after a decomposition step). Such results immediately imply (by standard methods) an \emph{additive} error bound $\|L_H - p L_{\Hf}\|_2 \lesssim \epsilon p \|L_{\Hf}\|_2$. 

But this is too weak to obtain a spectral approximation guarantee because $L_{\Hf}$ has non-zero eigenvalues on the order of $1/\ell \|L_{\Hf}\|_2$ (see \Cref{prop:eigs-johnson}; see also the discussion in~\cite{GHKM23}). The point is that additive control at the top-eigenvalue scale is too coarse for the lower Johnson eigenspaces. Thus, one of the sparsification guarantees above is a significant strengthening of the additive error estimate above, and it easily implies that estimate. Proving such sparsification analogs of the bounds of~\cite{GuruswamiKM22,HsiehKM23} would immediately resolve the question of finding algorithms for solving semirandom planted CSPs with guarantees matching those of the best known refutation algorithms. 

\paragraph{Correlated vs Independent Sampling:} One can view $K_H$ as the product of \emph{correlated sampling} of edges of $K_{\Hf}$ where whenever we choose to keep an edge from $K_C$, we must keep \emph{all} the edges. In the setting where we \emph{independently sample} the edges of $K_{\Hf}$ (a random variable with entropy ${n \choose \ell} \sim H(p) \cdot \ell^r n^r$), a standard application of matrix concentration inequalities analogous to~\cite{SpielmanSrivastava11} shows that a subsampling rate ensures an average degree of $\Omega(\log N)$ (for $N= {n \choose \ell})$ in the sampled graph suffices for a spectral approximation. \Cref{thm:main} shows that the same sampling rate suffices even in the correlated sampling setting (a random process with a significantly lower entropy of $\sim H(p) \cdot n^{2r}$). 
inequalities~\cite{SpielmanSrivastava11}.

\paragraph{Connections to Cayley and Schreier Graph Sparsification:} Kikuchi graphs are well-understood to be projections of the Cayley graph on $\mathbb{F}_2^n$  with generators $H$ (view each $C$ as its characteristic vectors in $F_2^n$). They are also Schreier graphs of $S_n$ on ${[n] \choose \ell}$. Thus, it is natural to try to apply the recent results~\cite{BasuKothariLiuMeka25,HsiehLeeMohantyPuttermanZhang25} on Cayley/Schreier graph sparsification. Unfortunately, the strongest known bound~\cite{BasuKothariLiuMeka25} obtained by these results yields a bound on $p$ that is larger by a factor $\sim \ell^3$. This loss is prohibitive in applications. Indeed, our application to planted $4$-XOR is trivialized if the loss factor exceeds $\sim \ell$.

\subsection{Application to Solving Planted $2r$-XOR} In the planted $2r$-XOR problem, the input is a random hypergraph $H$ with density $p$ together with independent assignments $b_C \in \{\pm 1\}$ such that, for some $x^* \in \{ \pm 1\}^n$, each $b_C$ equals $\prod_{i \in C} x_i^*$ with probability $1/2 + \eta$. This data naturally describes a collection of $2r$-XOR constraints generated as follows. Choose an $x^* \in \{\pm 1\}^n$. Consider the system of satisfiable (by $x^*$) $2r$-XOR constraints $\{\prod_{i \in C} x_i = x^*_C\}_{C \in H}$. The input is obtained by flipping each right hand side independently with probability $1/2-\eta$. The goal is to compute the hidden planted assignment $x^*$. 

By a connection first formulated in~\cite{GHKM23}, \Cref{thm:main} implies that the degree-$2\ell$ sum-of-squares SDP relaxation of the Max $2r$-XOR problem is \emph{exact}. That is, the optimal solution to the SDP relaxation is a rank $1$ matrix supported only on the true assignment $x^*$. 

\begin{corollary}[Algorithms for Planted $2r$-XOR] \label{cor:planted-xor}
Let $r \leq \ell \leq n/2$ be positive integers and let $0<\eta<1/2$. Let $H$ be a random $2r$-uniform hypergraph on $[n]$ with density $p$. Suppose there is an $x^* \in \{ \pm 1 \}^n$ such that for every $C \in H$, $b_C$ is independent and takes the value $\prod_{i \in C} x^*_i$ with probability $1/2+ \eta$. If $p$ is such that $p \cdot {n \choose 2r} \gtrsim_r \frac{1}{\eta^2(1/2-\eta)} n (\frac{n}{\ell})^{r-1}\log n$, then, with high probability, the natural $2\ell$-degree sum-of-squares relaxation is \emph{exact}. In particular, a trivial rounding recovers $x^*$ up to a global sign.
\end{corollary}

An algorithm with similar density running time trade-off was recently obtained by~\cite{BasuHsiehLinManohar25}. The density vs running time trade-off of their algorithm is similar except for a worse polynomial dependence on $\eta$ as opposed to the optimal $\eta^{-2}$ in the result above. A later work of Mao~\cite{Mao26NoisyKXOR} obtained an algorithm that needs a density that is smaller by a factor of $\log n$ (when $\ell$ grows superlogarithmically) in the bound above along the right quadratic dependence on $\eta$. The analysis in~\cite{BasuHsiehLinManohar25} relies on extracting an \emph{approximate} solution by rounding the sum-of-squares SDP relaxation for Max $2r$-XOR and then a certain local fixing step in order to recover the planted assignment. Our result, in particular, implies that the local-fixing step in their algorithm (for even-arity XOR) is unnecessary. We note that in their work~\cite{BasuHsiehLinManohar25}, the authors also obtain results for odd-arity XOR. Although we expect that a result analogous to~\Cref{cor:planted-xor} should hold for odd-arity XOR, it would require proving a suitable odd-uniformity analog of~\Cref{thm:main}. 

\begin{remark}[Odd $k$]
One may ask for analogous sparsification results for hypergraphs of odd uniformity. While the above definition of Kikuchi graphs does not make sense for odd arity, there are several natural analogs used in prior works~\cite{GuruswamiKM22,HsiehKM23,AlrabiahGKM22,HsiehKothariMohantyMunhaSudakov24,KothariManohar24LinearLCC}. In these constructions, each edge in the Kikuchi graph corresponds to \emph{two} hyperedges of the original hypergraph. Our techniques should extend analogously to this odd uniformity setting but we do not include this extension in the present version.
\end{remark}

\subsection{Proof Idea}
As we discussed above, it is not difficult to get an \emph{additive} error bound on $\|L_H - pL_{\Hf}\|_2 \lesssim \epsilon p\|L_{\Hf}\|_2$ but this bound is too large to be useful because $L_{\Hf}$ has a condition number of $\Omega(\ell)$ (or, there are non-zero eigenvalues of $L_{\Hf}$ that are smaller by a factor of $\ell$ compared to the largest). A natural route would be to work with normalized Laplacians $\tilde{L}_H = \sum_{C \in H} \tilde{L}_C$ where $\tilde{L}_C = (p L_{\Hf})^{-1/2} L_C (p L_{\Hf})^{-1/2}$ and the inverse square root is taken on the subspace orthogonal to $\1$, in which case an additive error bound would suffice. Unfortunately, for this setting, we do not know how to obtain a good estimate on the spectral norm of the ``variance term" that controls the error in the application of the matrix Bernstein inequality. 

Our approach is based on two simple but crucial (and somewhat miraculous, in the opinion of the author) observations about the relationship between the eigenvalues and eigenspace dimensions of $L_{\Hf}$ and the interaction of \emph{arbitrary} Kikuchi graphs with the eigenspaces of $L_{\Hf}$. Let us briefly describe both of these below. 

\paragraph{Eigenvalue Decay in $L_{\Hf}$:}
By standard results about the structure of the eigenspaces of $L_{\Hf}$, we know that for each $\ell  \geq i \geq 0$, there is an eigenspace $E_i$ of dimension ${n \choose i} -{n\choose {i-1}}$ with eigenvalue that roughly scales as $(i/\ell)\|L_{\Hf}\|_2$. When restricted to a single eigenspace, since all the eigenvalues of $L_{\Hf}$ are equal, $L_{\Hf}$ has condition number $1$, so one could attempt to apply matrix Bernstein analogously to our first failed attempt above. The dimension-dependent loss in the matrix Bernstein, when applied to $E_i$, is proportional to $\sqrt{ \log \dim(E_i)}$. It turns out that this improved bound is precisely enough for us to obtain the precise additive error bound that implies the required multiplicative error when restricted to $E_i$ at the required density $p$. This is based on the following useful coincidence (or, perhaps, a miraculous conspiracy): at the target density,
\[
\log \dim(E_i) \asymp i\log n
\qquad\text{and}\qquad
\lambda_i(pL_{\Hf}) \asymp_k i\log n/\epsilon^2.
\]
That is, the logarithm of the dimension of $E_i$ scales as the eigenvalue of $p L_{\Hf}$ on $E_i$. 

By a union bound on the $\ell+1$ different diagonal blocks, we can conclude that $\Pi_{E_i} L_H \Pi_{E_i} \approx \Pi_{E_i} L_{\Hf} \Pi_{E_i}$ for every $0 \leq i \leq \ell$. This does not, on its own, complete the proof of \Cref{thm:main}. This is because $E_i$s need not be (and in fact are not, in general) invariant subspaces for $L_{H}$. Thus, we are worried that for a function $f$ on ${[n] \choose \ell}$ that has projections in multiple $E_i$s, the quadratic forms of $L_H$ and $p L_{\Hf}$ may not be close.  

\paragraph{Band Limited Action of Kikuchi Graphs on $E_i$:} The crucial second idea that rescues our proof approach (with one modification) is the observation that \emph{every} Kikuchi graph satisfies a \emph{band locality} property in its action on $E_i$s. Specifically, for any $f \in E_i$ and every $L_C$, one has $L_C f \in \oplus_{i-k \leq j \leq i+k} E_j$, where $k=2r$. That is, $L_C$ sends an $f \in E_i$ to only ``nearby eigenspaces". In particular, $f^{\top} L_C g = 0$ whenever $f$ and $g$ are elements of $E_i$ and $E_j$ such that $|i-j|>k$. We prove this band locality property in \Cref{cor:band-limited-action-laplacian}.

This band locality does not imply that the per-eigenspace approximation argued above is enough. But if we were to upgrade the previous argument to obtain a multiplicative approximation on $\oplus_{i-k \leq j \leq i+k} E_j$ for every $j$ (instead of single $E_j$s), it turns out that we can finish the proof using the band locality property above. This stronger "per-$2k$-block of eigenspaces" guarantee also holds because the eigenvalues of $L_{\Hf}$ on any $E_j$ for $j \in [i-k,i+k] \cap \{1,2,\ldots,\ell\}$ are within an $O_k(1)$ factor of each other.  

Our argument with these two ideas is quite simple and presented in full in \Cref{sec:kikuchi-sparsification-even}. 
After preliminaries, we prove the planted-XOR application, then record the Johnson-scheme eigenspace facts and band locality, and finally prove \Cref{thm:main}.

\section{Preliminaries}
We will write $a \gtrsim b$ to denote $a \geq cb$ for some absolute constant $c>0$. Similarly, we will write $a \lesssim b$ to denote $a \leq cb$ for some absolute constant $c>0$. We will use $\gtrsim_k$ and $\lesssim_k$ to denote the versions where the hidden constant is allowed to depend on the parameter $k$. 

For two positive semidefinite matrices $A,B$, we will write $A \approx_{\delta} B$ as a shorthand for $(1+\delta)B \succeq A\succeq (1-\delta) B$. 

For any positive integer $k$, we will write $H^{(k)} \subseteq {[n] \choose k}$ for a $k$-uniform hypergraph on vertex set $[n]$. When the uniformity is clear from the context, we will suppress the superscript. 

We will say that \emph{$H^{(k)}$ is a random hypergraph with density $p$} if every possible hyperedge in ${n \choose k}$ is included independently in $H^{(k)}$ with probability $p$.

We recall the matrix Bernstein inequality~\cite{Tropp15}.

\begin{fact}
There is an absolute constant $c>0$ such that for every collection of independent random $N \times N$ matrices $A_1, A_2, \ldots$ with $\E[ A_i] = 0$ and $\|A_i\|_2 \leq R$ for each $i$,
\[
\Pr[ \|\sum_i A_i \|_2 \geq t ] \leq 2Q e^{-\frac{t^2}{2(\sigma^2 + Rt/3)}},
\]
where $\sigma^2 = \max\left\{ \|\E\sum_i A_i A_i^{\top}\|_2, \|\E\sum_i A_i^{\top} A_i\|_2  \right \}$ and $\E\sum_i A_i A_i^{\top}$ has rank $Q$.
\end{fact}

\section{Algorithms for Planted CSPs from Kikuchi Sparsification}

\begin{theorem}[Solving Random Planted $2r$-XOR] \label{thm:main-planted-xor-even}
Let $k=2r$ be an even positive integer, let $0<\eta<1/2$, and let $r\leq \ell = \ell(n) \leq n/2$. There is an algorithm that takes input a hypergraph $H \subseteq {[n] \choose k}$ and a collection of bits $\{b_C\}_{C \in H}$, runs in time $n^{O(\ell)}$ and outputs an assignment $x\in \{\pm 1\}^n$ with the following guarantee:
Suppose $H \subseteq {[n] \choose k}$ is chosen by including each hyperedge independently with probability $p$ such that $p {n \choose k} \gtrsim_k \frac{n}{\eta^2(1/2-\eta)} (n/\ell)^{k/2-1} \log n$. Let $\{b_C\}_{C \in H}$ be chosen by selecting an arbitrary $x^* \in \{\pm 1\}^n$ and setting $b_C = \prod_{i \in C} x^*_i$ with probability $\frac{1}{2} + \eta$ and  $b_C = -\prod_{i \in C} x^*_i$ with probability $\frac{1}{2} -\eta$ independently for each $C \in H$. Then, with probability at least $0.99$ over the draw of $H$ and $\{ b_C \}_{C \in H}$, the algorithm outputs $x \in \{x^*,-x^*\}$. 
\end{theorem}

Our algorithm is based on solving an SDP that relaxes a certain ``lifted" Max 2-XOR problem.

\paragraph{Algorithm:} 
The SDP variable is indexed by $\ell$-subsets; an integral rank-one solution associated with an assignment $x\in\{\pm 1\}^n$ has entries $X(S,T)=x_Sx_T$, where $x_S=\prod_{i\in S}x_i$.
Let $A_H = \sum_{C \in H} b_C K_C$ be the \emph{signed} Kikuchi graph of $H$ at level $\ell$. Our algorithm computes:
\begin{equation} \label{eq:sdp-even}
    \argmax_{X: X \succeq 0, \diag(X) = \1} \Tr(A_H \cdot X)
\end{equation} This is a semidefinite program and can be solved\footnote{Technically, the ellipsoid algorithm allows computing only an additive approximation to $X$ with $\exp(-n)$ error in $n^{O(\ell)}$ time. For our purposes, it suffices to round each entry of such an $X$ to the nearest integer.} in $n^{O(\ell)}$ time. We note that this SDP has a subset of constraints that appear in the canonical degree-$2\ell$ sum-of-squares relaxation for Max $2r$-XOR.

\textbf{Rounding}: Set $x_1 = 1$. For each $i>1$, choose any $S$ of size $\ell$ such that $i \in S$ and $1 \not \in S$. Set  $x_i = X(S,S \Delta \{1,i\})$.
For a rank-one moment matrix, this value is independent of the choice of $S$.

\begin{proposition} \label{prop:rounding-works-even}
Suppose $X = (y^{\odot \ell}) (y^{\odot \ell})^{\top}$ for some $y \in \{\pm 1\}^n$. Then, the rounding process above outputs either $y$ or $-y$. 
\end{proposition}
\begin{proof}
Let $x$ be the output of the rounding algorithm.
Suppose WLOG that $y_1 = 1$. We will prove that $x = y$. 
For any $S$ containing $i>1$ and not containing $1$, observe that $X(S,S \Delta \{i,1\}) = \prod_{j \in S} y_j \prod_{j \in S \Delta \{1,i\}} y_j = y_1 y_i = y_i$. This finishes the proof.
\end{proof}



\begin{proof}[Proof of Theorem~\ref{thm:main-planted-xor-even}]
Given Proposition~\ref{prop:rounding-works-even}, it suffices to prove that the output $X$ in \eqref{eq:sdp-even} equals $(x^*)^{\odot \ell}((x^*)^{\odot \ell})^{\top}$. Let $X' = X \star (x^*)^{\odot \ell}((x^*)^{\odot \ell})^{\top}$ where $\star$ denotes the entrywise/Hadamard product of matrices. Observe that $X'$ also has diagonal entries equal to $1$ and that by the Schur product theorem, $X'$ is positive semidefinite. This effectively allows us to pretend, for the sake of analysis, that the planted assignment is all ones. Our goal is equivalent to proving that $X' = \1 \1^{\top}$. 

Let $H_g \subseteq H$ be the set of all ``good" or error-free hyperedges $C$ such that $b_C = \prod_{i \in C} x^*_i$. Let $H_e = H \setminus H_g$ be the set of ``erroneous" hyperedges.  Then, $A_H = (x^*)^{\odot \ell}((x^*)^{\odot \ell})^{\top} \star (K_{H_g} - K_{H_e})$ where $\star$ denotes the \emph{Hadamard} (entrywise) product. We thus have $\Tr((x^*)^{\odot \ell}((x^*)^{\odot \ell})^{\top} \cdot A_H) = \Tr(D_{H_g}) - \Tr(D_{H_e})$ where $D_{H_g}$ and $D_{H_e}$ are the diagonal degree matrices of $K_{H_g}$ and $K_{H_e}$, respectively. It is thus enough to prove that whenever $X' \neq \1 \1^{\top}$, 
\[
\Tr(X \cdot A_H) = \Tr( X \cdot (x^*)^{\odot \ell}((x^*)^{\odot \ell})^{\top} \star (K_{H_g} - K_{H_e})) = \Tr(X' \cdot (K_{H_g} - K_{H_e})) < \Tr(D_{H_g}) - \Tr(D_{H_e}).
\]
Rearranging, this is equivalent to proving that:
\[
\Tr( X' \cdot L_{H_e}) <  \Tr(X' \cdot L_{H_g}),
\]
where $L_{H_g} = D_{H_g} - K_{H_g}$ is the Laplacian of $K_{H_g}$ and $L_{H_e} = D_{H_e} - K_{H_e}$ is the Laplacian of $K_{H_e}$.

Since $H$ is a random hypergraph with density $p$, $H_e$ is a random hypergraph with density $p (1/2 - \eta)$ and $H_g$ is a random hypergraph with density $p (1/2 + \eta)$. Let $\delta>0$ be a sufficiently small constant multiple of $\eta$. By the density assumption, both $H_e$ and $H_g$ satisfy the hypothesis of Theorem~\ref{thm:main-restated} with error parameter $\delta$. Applying Theorem~\ref{thm:main-restated} to each of $H_e$ and $H_g$, we have that with probability at least $1-1/(100n)$, 
\[
L_{H_e} \approx_\delta p (1/2 - \eta) L_{\Hf}
\] 
and 
\[
L_{H_g} \approx_\delta p (1/2 + \eta) L_{\Hf}.
\]
By a union bound, both of the above events occur with probability at least $0.98$. Condition on this event for the rest of the proof. 

Then, $L_{H_g} \succeq (1-\delta) p (1/2 + \eta) L_{\Hf}$ and $L_{H_e} \preceq (1+ \delta) p (1/2 - \eta) L_{\Hf}$. If $\delta \lesssim \eta$, then, $(1+ \delta) p (1/2 - \eta) < (1-\delta) p (1/2 + \eta)$. 

Now, as we recall (from the standard eigenvalue formulas for Johnson scheme in Proposition~\ref{prop:eigs-johnson}) $L_{\Hf}$ has a one-dimensional kernel (along the all-ones direction). Indeed, under $r\leq \ell \leq n/2$, \Cref{prop:eigs-johnson} gives $\lambda_i>0$ for every $i\geq 1$. Let $\Pi$ be the projection matrix orthogonal to the all-ones direction. Then, we must have:
\[
\Pi L_{H_e} \Pi \prec \Pi L_{H_g} \Pi.
\]

Now, observe that if $X' \neq \1 \1^{\top}$, $\Pi X' \Pi$ is non-zero. Indeed, a positive semidefinite matrix with diagonal entries equal to $1$ whose projection onto $\1^{\perp}$ vanishes must equal $\1\1^{\top}$. Thus, taking trace products with $X'$ on both sides  yields: 
\[
\Tr(X' \cdot L_{H_e}) = \Tr( X' \cdot \Pi L_{H_e} \Pi ) < \Tr ( X' \cdot \Pi L_{H_g} \Pi) = \Tr(X' \cdot L_{H_g}).
\]
This completes the proof. 
\end{proof}

\section{Johnson Eigenspaces and Band Limited Action of Kikuchi Graphs} \label{sec:johnson-band-limited}
A matrix $M$ with rows and columns indexed by ${[n] \choose \ell}$ is set-symmetric if its entries depend only on $|S \Delta T|$. Such matrices form a commutative algebra called the \emph{Johnson scheme} and thus admit a simultaneous eigenspace decomposition. We work throughout this section in the regime $\ell \leq n/2$. We will need some elementary facts about these eigenspaces and the corresponding eigenvalues. We direct the reader to~\cite{Godsil93} for detailed background on eigenvalues and eigenspaces of the Johnson scheme. We will view these matrices as acting on the linear space of all functions $f:{[n] \choose \ell} \rightarrow \R$. Let us define some special functions that will help us in the following.  
For \(R\subset[n]\), define $p_R$ as the function such that at any $S$, $p_R(S)=\one_{\{R\subset S\}}$.
For \(d\ge 0\), define the space of polynomials of degree at most $d$ by $F_d=\operatorname{span}\{p_R: |R|\le d\}$. We also set \(F_{-1}=0\), and once \(d\ge \ell\) we identify \(F_d\) with \(V\). We observe that $F_d$ is spanned by ``homogeneous" degree $d$ polynomials. For any $a \in [n]$, let $x_a = \one_{(a \in S)}$ be the indicator of whether a given input $S$ contains the element $a$. Then, observe that degree $d$ polynomials are in fact linear combinations of $\prod_{a \in R} x_a$ for $|R| = d$. 

\begin{lemma}
For any $d \leq \ell$, $F_d=\operatorname{span}\{p_R: |R|=d\}$ and
$\dim F_d=\binom nd$. In particular, \(F_\ell=V\).
\end{lemma}

\begin{proof}
We first prove that $\{p_R | |R|=d\}$ are linearly independent on ${[n] \choose \ell}$ by induction on $d$.

For \(d=0\), the claim is clear. Suppose \(d\ge 1\), and assume $\sum_{|R|=d} a_R p_R(S)=0$
for every \(S\in\X\). Fix distinct \(p,q\in[n]\), and let $Y\subset [n]\setminus\{p,q\}$ such that $|Y|=\ell-1$.
Comparing the above relation at \(Y\cup\{p\}\) and at \(Y\cup\{q\}\), we get $\sum_{\substack{U\subset Y\\ |U|=d-1}}
\bigl(a_{U\cup\{p\}}-a_{U\cup\{q\}}\bigr)=0$
for every such \(Y\).

This is a relation among the functions \(p_U\), \(|U|=d-1\), on the \((\ell-1)\)-subsets of the ground set \([n]\setminus\{p,q\}\). Since
$d-1\le \min\{\ell-1,n-\ell-1\}$,
the induction hypothesis gives $a_{U\cup\{p\}}=a_{U\cup\{q\}}$ for every \(U\subset[n]\setminus\{p,q\}\) of size \(d-1\). Hence any two \(d\)-sets differing in exactly one element have the same coefficient. The graph on \(d\)-subsets where two sets are adjacent if they differ in one element is connected, so all \(a_R\) are equal to some constant \(a\). Evaluating the original relation at any \(S\in\X\) gives $0=a\binom{\ell}{d}$.
Thus \(a=0\). This proves linear independence.

Now let \(|Q|<d\). On \(\X\) we have
\[
p_Q(S)
=
\frac{1}{\binom{\ell-|Q|}{d-|Q|}}
\sum_{\substack{R\supset Q\\ |R|=d}} p_R(S).
\]
Indeed, if \(Q\subset S\), then the number of \(d\)-sets \(R\) with $Q\subset R\subset S$
is \(\binom{\ell-|Q|}{d-|Q|}\), while if \(Q\not\subset S\), both sides vanish. Therefore \(F_d\) is spanned by the functions \(p_R\) with \(|R|=d\). Since these functions are linearly independent, $\dim F_d=\binom nd$.
\end{proof}

\begin{lemma}
For \(0\le i\le \ell\), define $E_i=F_i\cap F_{i-1}^{\perp}$.
Then
\[
V=E_0\oplus E_1\oplus\cdots\oplus E_\ell
\]
orthogonally, and $\dim E_i=\binom ni-\binom n{i-1}$.
\end{lemma}

\begin{proof}
Since $0=F_{-1}\subset F_0\subset F_1\subset\cdots\subset F_\ell=V$,
we have $F_i=F_{i-1}\oplus E_i$
orthogonally. Iterating gives the orthogonal decomposition of \(V\). Also, $\dim E_i
=
\dim F_i-\dim F_{i-1}
=
\binom ni-\binom n{i-1}$.
\end{proof}

These spaces are the standard Johnson eigenspaces.
The eigenspaces of the matrices from the Johnson scheme can be described in terms of the more familiar polynomial functions. We record this as the following easy observation.  
\begin{lemma}[Polynomial description] \label{lem:poly-description}
Let \(x_a(S)=\one_{\{a\in S\}}\) and let $f:{[n] \choose \ell} \rightarrow \R$ be a degree $\leq d$ polynomial in $x_{a}$. Then, $f \in F_d$. 
\end{lemma}

\begin{proof}
Every monomial $x_{a_1}x_{a_2}\cdots x_{a_q}$
with distinct variables is exactly \(p_R\), where $R=\{a_1,\dots,a_q\}$.
Thus every polynomial of degree at most \(d\) is a linear combination of functions \(p_R\) with \(|R|\le d\), hence lies in \(F_d\).
\end{proof}

\begin{proposition}[Eigenvalues]
Assume $r\leq \ell \leq n/2$. Let $L = D_k - K_{H_{full}}$ be the Laplacian of $K_{H_{full}}$, where $D_k = d_k I$ and $d_k = {\ell \choose r}{{n-\ell} \choose r}$ is the degree of any vertex in $K_{H_{full}}$. By the Eberlein polynomial formula for the Johnson scheme~\cite[Theorem~3.13]{AalipourEtAl16}, the eigenvalue $\lambda_i$ on $E_i$ is given by: 
\[
\lambda_i
=
\binom{\ell}{r}\binom{n-\ell}{r}
-
\sum_{j=0}^{r}
(-1)^j
\binom ij
\binom{\ell-i}{r-j}
\binom{n-\ell-i}{r-j}.
\]

For \(r\) fixed and \(\ell\to\infty\), we have:
\[
\frac{\lambda_i}{d_k}
=
1-
\left(1-\frac{i}{\ell}\right)^{k/2}
\left(1-\frac{i}{n-\ell}\right)^{k/2}
+
O_k\left(\frac1\ell\right).
\]

In particular, for fixed $k=2r$ and $n$ sufficiently large depending only on $k$, for every $1\leq i\leq \ell$,
\[
c_1 \frac{i}{\ell} \leq \frac{\lambda_i}{d_k} \leq c_2 \frac{i}{\ell}
\]
for some constants $c_1,c_2>0$ depending only on $k$. 
Indeed, for $i\leq c_k\ell$, expanding the exact formula to first order gives $\lambda_i/d_k=\Theta_k(i(1/\ell+1/(n-\ell)))=\Theta_k(i/\ell)$. For $i>c_k\ell$, the displayed approximation gives $\lambda_i/d_k=\Theta_k(1)=\Theta_k(i/\ell)$.

\label{prop:eigs-johnson}
\end{proposition}

\subsection{Band Limited Action of $K_C$s}
Kikuchi graphs of arbitrary hypergraphs are \emph{not} matrices from the Johnson scheme. Thus, $E_i$s are no longer invariant subspaces for the action of arbitrary Kikuchi graphs. Nevertheless, we make\footnote{While the proof is simple and short, we are not aware of any work that has made or used this observation in prior work.} the following simple observation that the action of Kikuchi graphs on $E_i$ is more structured than one might expect. 
\begin{lemma}[Lower Limits from Upper Limits] \label{lem:lower-from-upper}
Let \(T:V\to V\) be self-adjoint. Suppose that for some \(s\ge 0\), $T F_d\subseteq F_{d+s}$ for every \(d\). Then, for every \(0\le i\le \ell\),
\[
T E_i
\subseteq
\bigoplus_{j=\max(0,i-s)}^{\min(\ell,i+s)} E_j.
\]
\end{lemma}

\begin{proof}
Since \(E_i\subset F_i\), the assumption gives
$T E_i\subseteq F_{i+s} = \bigoplus_{j=0}^{\min(\ell,i+s)} E_j$.
We now prove the lower containment.

Let \(x\in E_i\), and let \(y\in E_j\) with \(j<i-s\). Then
$T y\in F_{j+s}\subseteq F_{i-1}$.
Since \(x\perp F_{i-1}\), and since \(T\) is self-adjoint,
$\langle T x,y\rangle
=
\langle x,T y\rangle
=
0$.
Therefore \(T x\) has no component in any \(E_j\) with \(j<i-s\). This proves the claim.
\end{proof}

\begin{corollary}[$K_C$s have a band-limited action] \label{cor:band-limited-action}
For every \(0\le i\le \ell\),
\[
K_C E_i
\subseteq
\bigoplus_{j=\max(0,i-k)}^{\min(\ell,i+k)} E_j.
\]
\end{corollary}

\begin{proof}
$K_C$ is a symmetric matrix so by Lemma~\ref{lem:lower-from-upper} it is enough to prove that 
$K_C F_d\subseteq F_{d+k}$ 
for every \(d\). Since $F_d$ is spanned by \(p_R\) with \(|R|\le d\), we can restrict attention to the action of $K_C$ on the functions $p_R$. Write $x_a(S)=\one_{\{a\in S\}}$.
Then
$p_R(S\triangle C)
=
\prod_{a\in R\setminus C} x_a
\prod_{a\in R\cap C} (1-x_a)$.
This is a polynomial in the variables \(x_a\) of degree at most \(|R|\le d\).

Also,
\[
\one_{\{|S\cap C|=r\}}
=
\sum_{\substack{U\subset C\\ |U|=r}}
\prod_{u\in U} x_u
\prod_{v\in C\setminus U} (1-x_v).
\]
This is a polynomial of degree at most \(|C|=k\). Therefore, 
$(K_Cp_R)(S) = \one_{\{|S\cap C|=r\}}p_R(S\triangle C)$
is represented by a polynomial of degree at most \(d+k\). Applying Lemma \ref{lem:poly-description} finishes the proof. 
\end{proof}

\begin{corollary}[$L_C$s have a band-limited action] \label{cor:band-limited-action-laplacian}
For every \(0\le i\le \ell\),
\[
L_C E_i
\subseteq
\bigoplus_{j=\max(0,i-k)}^{\min(\ell,i+k)} E_j.
\]
\end{corollary}

\begin{proof}
Write $L_C = D_C - K_C$, where $D_C$ is the diagonal degree matrix of $K_C$. The $K_C$ term has the desired property by \Cref{cor:band-limited-action}. It remains to handle $D_C$. Since
\[
(D_C p_R)(S)=\one_{\{|S\cap C|=r\}}p_R(S),
\]
and $\one_{\{|S\cap C|=r\}}$ is a degree-$k$ polynomial in the variables $x_a$, we have $D_C F_d \subseteq F_{d+k}$ for every $d$. Applying \Cref{lem:lower-from-upper} to the self-adjoint operator $D_C$ gives the same band-limited conclusion for $D_C$, and therefore also for $L_C$.
\end{proof}

\section{Kikuchi Graphs of Random Hypergraphs $\approx$ Johnson} \label{sec:kikuchi-sparsification-even}
In this section, we prove \Cref{thm:main}.

\begin{theorem} [\Cref{thm:main} restated] \label{thm:main-restated}
Let $k=2r$ for a positive integer $r$, let $0<\epsilon\leq 1/2$, and let $H$ be a random $k$-uniform hypergraph with density $p$. For any $r\leq\ell = \ell(n) \leq n/2$, let $L_H = D_H - K_H$ be the Laplacian matrix of $K_H = K_H^{(\ell)}$. If $p$ is such that $p {n\choose k} \gtrsim_k (n/\epsilon^2) (n/\ell)^{k/2-1}\log n$, then, with probability at least $1-1/(100n)$ over the draw of $H$, 
\[
(1-\epsilon) p L_{\Hf} \preceq L_{H} \preceq (1+\epsilon) p L_{\Hf}.
\]
\end{theorem}
\begin{proof}
For each $0 \leq i \leq \ell$, let
\[
W_i=\bigoplus_{h=\max(0,i-k)}^{\min(\ell,i+k)}E_h,
\qquad
W_i^+=W_i\cap \1^\perp,
\]
and let $\Pi_i$ and $\Pi_{W_i}$ be the orthogonal projection matrices onto $W_i^+$ and $W_i$, respectively. Let $\Pi_{E_i}$ be the orthogonal projection matrix onto the subspace $E_i$. Let $\lambda_h$ denote the eigenvalue of $L_{\Hf}$ on $E_h$, and define the positive block scales
\[
\Lambda_i^-=\min_{\substack{h\in[\max(0,i-k),\min(\ell,i+k)]\\ h\geq 1}}p\lambda_h,
\qquad
\Lambda_i^+=\max_{\substack{h\in[\max(0,i-k),\min(\ell,i+k)]\\ h\geq 1}}p\lambda_h.
\]
All Laplacians vanish on the all-ones direction, so the $E_0$ component never contributes; this is why the concentration argument below is stated on $W_i^+$.

We will prove the following claim, which shows that the sparsification claim holds if we project the matrices down to any $W_i^+$. The proof relies on a (rank-sensitive) application of matrix Bernstein inequality and exploits the crucial fact that the eigenvalue of $L_{\Hf}$ on $E_i$ is proportional to $\log \dim(E_i)$:

\begin{claim} \label{claim:subspace-bernstein-even}
Let $\epsilon' \lesssim_k \epsilon$. Suppose $p{n\choose k} \gtrsim_k (n/\epsilon'^2) (n/\ell)^{k/2-1} \log n$. Then, with probability at least $1-1/(100n)$, for every $0\leq i\leq \ell$,
\[
\left\|\Pi_i(L_H-pL_{\Hf})\Pi_i\right\|_2\leq \epsilon'\Lambda_i^-.
\]
Consequently,
\[
\Pi_{W_i} L_H \Pi_{W_i}  \approx_{\epsilon'} p\Pi_{W_i}L_{\Hf}\Pi_{W_i}.
\]
\end{claim}
While not true in general, we will exploit the band limited action of Kikuchi graphs on $E_i$ to show that this is sufficient to imply the theorem with error $O_k(\epsilon')$.

Condition on $L_H$ satisfying the conclusion of the claim above. Let $A=L_H-pL_{\Hf}$ and let $f$ be an arbitrary function ${[n] \choose \ell} \rightarrow \R$. We will prove $f^{\top} L_H f \in (1 \pm \epsilon) p f^{\top} L_{\Hf}f$. Let $f_i = \Pi_{E_i} f$. Then, $f = \sum_{i = 0}^{\ell} f_i$. Observe that since the $E_i$ are eigenspaces of $L_{\Hf}$, we have: 
\[
f^{\top} pL_{\Hf} f = \sum_i f_i^{\top} pL_{\Hf} f_i 
\]

From Corollary~\ref{cor:band-limited-action-laplacian}, we know that $f_i^{\top} L_H f_j = 0$ if $|i-j|>k$. Since $L_{\Hf}$ is diagonal in the $E_i$ decomposition, this implies that $f_i^{\top} A f_j=0$ if $|i-j|>k$. From Claim~\ref{claim:subspace-bernstein-even} applied to each $f_i$, we have: 
\[
  |f_i^{\top} A f_i| \leq \epsilon' f_i^{\top} pL_{\Hf} f_i.
\] 
Further, for any $i,j$ such that $|i-j|\leq k$, the cross term is zero if either $i=0$ or $j=0$ because $A$ kills constants. Otherwise, $f_i+f_j\in W_i^+$, so by Claim~\ref{claim:subspace-bernstein-even} and the above bound, we have:
\begin{align*}
|f_i^{\top} A f_j| &= \frac{1}{2}\left| (f_i + f_j)^{\top} A (f_i + f_j) - f_i^{\top} A f_i - f_j^{\top} A f_j\right|\\
&\leq \epsilon' (f_i^{\top} pL_{\Hf} f_i + f_j^{\top} pL_{\Hf} f_j).
\end{align*} 
Thus, we have:
\begin{align*}
&|f^{\top} A f| = |\sum_{0\leq i,j\leq \ell} f_i^{\top} A f_j|\\
&=|\sum_{0 \leq i \leq \ell} f_i^{\top} A f_i + \sum_{0 \leq i \leq \ell} \sum_{j \neq i: |i-j|\leq k} f_i^{\top} A f_j|\\
&\leq 5k \epsilon' \left(\sum_{i} f_i^{\top} pL_{\Hf} f_i \right). 
\end{align*}
Choosing $\epsilon' = \epsilon/5k$ finishes the proof.

\begin{proof}[Proof of claim]

Fix $0\leq i\leq \ell$ and write $Q_i=\dim(W_i^+)$. We know that $L_{H} = \sum_{C \in H_{full}} \1(C \in H) L_C$ where $L_C $ is the Laplacian of $K_C$ and $\1(C \in H)$ is the indicator random variable for whether $C \in H$. Thus, $\E[ \1(C \in H)] = p$ and
\[
X_C=(\1(C \in H) - p) \Pi_i L_C \Pi_i
\]
is a random matrix with mean $0$ and spectral norm at most $\|L_C\|_2 \leq 2$ with probability $1$. Further, using $L_C^2=2L_C$,
\begin{align*}
\sum_{C} \E[X_C^2]  &\preceq \sum_{C}p (\Pi_i L_C \Pi_i)^2 \\
&\preceq \sum_{C}p \Pi_i L_C^2 \Pi_i = 2p \Pi_i L_{\Hf} \Pi_i \preceq 2\Lambda_i^+ I.
\end{align*}
By \Cref{prop:eigs-johnson}, the eigenvalues inside a $k$-window satisfy $\Lambda_i^+\leq C_k\Lambda_i^-$. At the assumed density, after increasing the hidden constant, $\Lambda_i^-\geq C_k\epsilon'^{-2}\log(200n^2Q_i)$.

Matrix Bernstein therefore gives
\[
\Pr\left[\left\|\sum_C X_C\right\|_2\geq \epsilon'\Lambda_i^-\right]
\leq
2Q_i\exp\left(-\frac{\epsilon'^2(\Lambda_i^-)^2}{2(2\Lambda_i^+ + 2\epsilon'\Lambda_i^-/3)}\right)
\leq \frac{1}{200n^2}.
\]
Thus, with probability at least $1-1/(200n^2)$,
\[
\left\|\Pi_i(L_H-pL_{\Hf})\Pi_i\right\|_2\leq \epsilon'\Lambda_i^-.
\]
A union bound over the $\ell+1$ choices of $i$ completes the proof of the claim. \end{proof} \end{proof}

\section*{Acknowledgement}
We thank Venkatesan Guruswami, Tim Hsieh and Peter Manohar for the many meetings and discussions about the semirandom Planted k-XOR problem together over the last few years. We thank Arpon Basu, Josh Brakensiek, Venkatesan Guruswami, Tim Hsieh, Andrew Lin, Louie Putterman and Peter Manohar who read and provided valuable comments on previous versions of the manuscript over the last year.  

\bibliographystyle{alpha}
\bibliography{references}

\appendix

\end{document}